\documentclass{article}
\usepackage{amsmath}
\usepackage{amsthm}
\usepackage{amssymb}
\usepackage{url}
\usepackage{epsfig}
\usepackage{latexsym}

\newcommand{\SB}{\{\,}
\newcommand{\SM}{\;{:}\;}
\newcommand{\SE}{\,\}}
\newcommand{\Card}[1]{|#1|}

\newcommand{\CCC}{\mathcal{C}}

\newcommand{\vars}{\text{\normalfont vars}}
\newcommand{\lits}{\text{\normalfont lits}}

\newcommand{\NP}{\text{\normalfont NP}}
\renewcommand{\P}{\text{\normalfont P}}

\newcommand{\XP}{\text{\normalfont XP}}

\newcommand{\Prob}{\text{\normalfont Prob}}
\newcommand{\dom}{\rightarrow}
\newcommand{\ep}{\epsilon}
\newcommand{\GF}{\text{\normalfont GF}}
\newcommand{\2}{\vspace{3mm}}

\newtheorem{lemma}{Lemma}
\newtheorem{theorem}{Theorem}

\newtheorem{proposition}{Proposition}
\newtheorem{krule}{Reduction Rule}

\let\phi=\varphi

\newcommand{\Case}[1]{\paragraph{#1}}

\begin{document}

\title{A Probabilistic Approach to Problems Parameterized Above or Below Tight Bounds\footnote{An extended abstract of this paper will appear in the Proceedings of IWPEC 2009.}}

\author{Gregory Gutin\footnote{Department of Computer Science,
Royal Holloway University of London, Egham, Surrey TW20 0EX,
England, UK, \texttt{gutin@cs.rhul.ac.uk}}\and Eun Jung Kim\footnote{Department of Computer Science,
Royal Holloway University of London, Egham, Surrey TW20 0EX,
England, UK, \texttt{eunjung@cs.rhul.ac.uk}}\and Stefan
  Szeider\footnote{Department of Computer Science, Durham University,
Durham DH1 3LE, England, UK,
\texttt{stefan@szeider.net}}\and Anders Yeo\footnote{Department of Computer Science,
Royal Holloway University of London, Egham, Surrey TW20 0EX,
England, UK, \texttt{anders@cs.rhul.ac.uk}}}

\date{ }
\maketitle

\begin{abstract}
  We introduce a new approach for establishing fixed-parameter
  tractability of problems parameterized above tight lower bounds or below
  tight upper bounds. To
  illustrate the approach we consider three problems of this type of
  unknown complexity that were introduced by Mahajan, Raman and Sikdar
  (J. Comput. Syst. Sci.  75, 2009). We show that a generalization of
  one of the problems and non-trivial special cases of the other two
  are fixed-parameter tractable.
 \end{abstract}

\pagenumbering{arabic}
\pagestyle{plain}

\section{Introduction}

A parameterized problem $\Pi$ can be considered as a set of pairs
$(I,k)$ where $I$ is the \emph{main part} and $k$ (usually an integer)
is the \emph{parameter}. $\Pi$ is called \emph{fixed-parameter
  tractable} (FPT) if membership of $(I,k)$ in $\Pi$ can be decided in
time $O(f(k)|I|^c)$, where $|I|$ denotes the size of $I$, $f(k)$ is a
computable function, and $c$ is a constant independent of $k$ and $I$
(for further background and terminology on parameterized complexity we
refer the reader to the
monographs~\cite{DowneyFellows99,FlumGrohe06,Niedermeier06}). If the
nonparameterized version of $\Pi$ (where $k$ is just a part of the input)
is $\NP$-hard, then the function $f(k)$ must be superpolynomial
provided $\P\neq \NP$. Often $f(k)$ is ``moderately exponential,''
which makes the problem practically feasible for small values of~$k$.
Thus, it is important to parameterize a problem in such a way that the
instances with small values of $k$ are of real interest.

Consider the following well-known problem: given a digraph $D=(V,A)$,
find an acyclic subdigraph of $D$ with the maximum number of arcs. We
can parameterize this problem ``naturally'' by asking whether $D$
contains an acyclic subdigraph with at least $k$ arcs. It is easy to
prove that this parameterized problem is fixed-parameter tractable by
observing that $D$ always has an acyclic subdigraph with at least
$|A|/2$ arcs. (Indeed, consider a bijection $\alpha:\ V \dom
\{1,\ldots ,|V|\}$ and the following subdigraphs of $D$: $(V,\SB xy\in
A \SM \alpha(x)< \alpha(y)\SE)$ and $(V,\SB xy\in A \SM \alpha(x)>
\alpha(y)\SE)$.  Both subdigraphs are acyclic and at least one of them
has at least $|A|/2$ arcs.) However, $k\le |A|/2$ for every small
value of $k$ and almost every practical value of $|A|$ and, thus, our
``natural'' parameterization is of almost no practical or theoretical
interest.


Instead, one should consider the following parameterized problem:
decide whether $D=(V,A)$ contains an acyclic subdigraph with at least
$|A|/2+k$ arcs. We choose $|A|/2+k$ because $|A|/2$ is a \emph{tight
  lower bound} on the size of a largest acyclic subdigraph. Indeed,
the size of a largest acyclic subdigraph of a symmetric digraph
$D=(V,A)$ is precisely $|A|/2$. (A digraph $D=(V,A)$ is {\em
  symmetric} if $xy\in A$ implies $yx\in A$.)

In a recent paper~\cite{MahajanRamanSikdar09} Mahajan, Raman and
Sikdar provided several examples of problems of this type and argued
that a natural parameterization is one above a tight lower bound for
maximization problems, and below a tight upper bound for minimization
problems. Furthermore, they observed that only a few non-trivial
results are known for problems parameterized above a tight lower
bound~\cite{GutinRafieySzeiderYeo07,GutinSzeiderYeo08,HeggernesPaulTelleVillanger07,MahajanRaman99},
and they listed several problems parameterized above a tight lower
bound whose complexity is unknown. The difficulty in showing whether
such a problem is fixed-parameter tractable can be illustrated by the
fact that often we even do not know whether the problem is in $\XP$,
i.e., can be solved in time $O(|I|^{g(k)})$ for a computable
function~$g(k)$. For example, it is non-trivial to see that the above-mentioned
digraph problem is in $\XP$ when parameterized above the
$|A|/2$ bound.

In this paper we introduce the \emph{Strictly Above/Below Expectation Method (SABEM)}, a novel
approach for establishing the fixed-parameter tractability of maximization
problems parameterized above tight lower bounds and minimization problems parameterized below tight upper bounds.
The new method is based on
probabilistic arguments and utilizes certain probabilistic inequalities. We will
state the equalities in the next section, and in the subsequent sections we will
apply SABEM to three open problems posed in \cite{MahajanRamanSikdar09}.

Now we give a very brief description of the new method with respect to a given
problem $\mathrm{\Pi}$ parameterized above a
tight lower bound or below a tight upper bound.
We first apply some reductions rules to reduce $\mathrm{\Pi}$
to its special case $\mathrm{\Pi}'.$ Then we introduce a random
variable $X$ such that the answer to $\mathrm{\Pi}$  is \textsc{yes} if and only if $X$ takes, with positive
probability, a value greater or equal to the parameter~$k$. Now using some probabilistic inequalities on $X$,
we derive upper bounds on the size of \textsc{no}-instances of $\mathrm{\Pi}'$ in
terms of a function of the parameter~$k$. If the size of a given instance
exceeds this bound, then we know the answer is \textsc{yes}; otherwise, we produce a \emph{problem
kernel}~\cite{DowneyFellows99}. In many cases, we obtain problem kernels of polynomial size.

In Section \ref{secLOP}, we consider the \textsc{Linear Ordering} problem, a
generalization of the problem discussed above: Given a digraph $D=(V,A)$ in
which each arc $ij$ has a positive integral weight $w_{ij}$, find an acyclic
subdigraph of $D$ of maximum weight.  Observe that $W/2$, where $W$ is the sum of
all arc weights, is a tight lower bound for \textsc{Linear Ordering}. We prove
that the problem parameterized above $W/2$ is fixed-parameter tractable and
admits a quadratic kernel. Note that this parameterized problem generalizes the
parameterized maximum acyclic subdigraph problem considered
in~\cite{MahajanRamanSikdar09}; thus, our result answers the corresponding open
question of~\cite{MahajanRamanSikdar09}.

In Section \ref{secLin}, we consider the problem \textsc{Max Lin-2}: Given a
system of $m$ linear equations $e_1,\dots,e_m$ in $n$ variables over $\GF(2)$,
and for each equation $e_j$ a positive integral weight $w_j$; find an assignment
of values to the $n$ variables that maximizes the total weight of the satisfied
equations. We will see that $W/2$, where $W=w_1+\cdots +w_m$, is a tight
lower bound for \hbox{\textsc{Max Lin-2}}. The complexity of the problem
parameterized above $W/2$ is open~\cite{MahajanRamanSikdar09}. We prove that the
following three special cases of the parameterized problem are fixed-parameter tractable:
(1) there is a set $U$ of variables such that each equation has an odd number of variables from $U$, (2) there is a constant $r$ such that
each equation involves at most $r$ variables, (3) there is a constant $\rho$
such that any variable appears in at most $\rho$ equations.  For all three cases we obtain kernels with $O(k^2)$ variables
and equations. We also show that if we allow the weights $w_j$ to be positive
reals, the problem is $\NP$-hard already if $k=1$ and each equation involves two
variables.

In Section \ref{secSAT}, we consider the problem \textsc{Max Exact $r$-SAT}:
given an exact $r$-CNF formula $\CCC$ with $m$ clauses (i.e., a CNF formula
where each clause contains exactly $r$ distinct literals), find a truth
assignment that satisfies the maximum number of clauses. Here a tight lower
bound is $(1-2^{-r})m$; the complexity of the problem parameterized above
$(1-2^{-r})m$ is an open question~\cite{MahajanRamanSikdar09}. This seems to be
the most difficult problem of the three considered. We obtain a quadratic kernel
for a non-trivial special case of this problem.

In Section \ref{secd}, we briefly mention minimization problems
parameterized below tight upper bounds,
provide further discussions of problems considered in this paper and
point out to a very recent result obtained using our new method.

\section{Probabilistic Inequalities}\label{sect:ineq}

In our approach we introduce a random variable $X$ such that the answer to the
problem parameterized above a tight lower bound or below a tight upper bound is {\sc yes} if and only if $X$
takes with positive probability a value greater or equal to the
parameter~$k$.

In this paper all random variables are real. A random variable is {\em discrete} if its distribution function
has a finite or countable number of positive increases.
A random variable $X$ is a {\em symmetric}  if $-X$ has the same distribution function as $X$.
If $X$ is discrete, then $X$ is symmetric if and only if $\Prob(X=a)=\Prob(X=-a)$
for each real $a.$  Let $X$ be a symmetric variable for which the first moment $\mathbb{E}(X)$ exists.
Then $\mathbb{E}(X)=\mathbb{E}(-X)=-\mathbb{E}(X)$ and, thus, $\mathbb{E}(X)=0.$
The following is easy to prove \cite{vovk}.
\begin{lemma}\label{eqsym}
If $X$ is a symmetric random variable and $\mathbb{E}(X^2)<\infty$, then
$$\Prob(\ X \ge \sqrt{\mathbb{E}(X^2)}\ )>0.$$\end{lemma}
See Sections \ref{secLOP} and \ref{secLin} for applications of Lemma \ref{eqsym}. Unfortunately,
often $X$ is not symmetric, but Lemma~\ref{lem32}
provides an inequality that can be used in many such cases. This lemma was
proved by Alon et al.~\cite{AlonGutinKrivelevich04}; a weaker
version was obtained by H{\aa}stad and Venkatesh~\cite{HastadVenkatesh02}.
\begin{lemma}\label{lem32}
  Let $X$ be a random variable and suppose that its first, second and forth
  moments satisfy $\mathbb{E}(X)=0$, $\mathbb{E}(X^2)=\sigma^2>0$ and
  $\mathbb{E}(X^4) \leq b \sigma^4$, respectively.  Then $\Prob(\ X >
  \frac{\sigma}{4 \sqrt b}\ ) \geq \frac{1}{4^{4/3} b}$.
\end{lemma}
Since it is often rather nontrivial to evaluate $\mathbb{E}(X^4)$ in order to
check whether $\mathbb{E}(X^4) \leq b \sigma^4$ holds, one can sometimes use the following extension
of Khinchin's Inequality by Bourgain~\cite{Bourgain80}.
\begin{lemma}\label{lem41}
  Let $f=f(x_1,\ldots,x_n)$ be a polynomial of degree $r$ in $n$ variables
  $x_1,\ldots,x_n$ with domain $\{-1,1\}$. Define a random variable $X$ by
  choosing a vector $(\ep_1,\ldots,\ep_n)\in \{-1,1\}^n$ uniformly at random and
  setting $X=f(\ep_1,\ldots,\ep_n)$.  Then, for every $p \geq 2$, there is a
  constant $c_p$ such that \[ (\mathbb{E}(|X|^p ))^{1/p} \leq (c_p)^r
  (\mathbb{E}(X^2))^{1/2}. \] In particular, $c_4 \leq 2^{3/2}$.
\end{lemma}

\section{Linear Ordering}\label{secLOP}

Let $D=(V,A)$ be a digraph with no loops or parallel arcs in which every arc
$ij$ has a positive weight $w_{ij}$. The problem of finding an acyclic
subdigraph of $D$ of maximum weight, known as \textsc{Linear Ordering}, has
applications in economics~\cite{BangjensenGutin09}. Let $n=|V|$ and consider a
bijection $\alpha: V \dom \{1,\ldots,n\}$. Observe that the subdigraphs $(V,\SB
ij\in A\SM \alpha(i)<\alpha(j)\SE)$ and $(V,\SB ij\in A\SM
\alpha(i)>\alpha(j)\SE)$ are acyclic. Since the two subdigraphs contain all arcs
of $D$, at least one of them has weight at least $W/2$, where $W=\sum_{ij\in
  A}w_{ij}$, the \emph{weight} of~$D$. Thus, $W/2$ is a lower bound on the
maximum weight of an acyclic subdigraph of~$D$.  Consider a digraph $D$ where
for every arc $ij$ of $D$ there is also an arc $ji$ of the same weight. Each
maximum weight subdigraph of $D$ has weight exactly $W/2$. Hence the lower bound
$W/2$ is tight.

\begin{quote}
  \textsc{Linear Ordering Above Tight Lower Bound} (\textsc{LOALB})\nopagebreak

  \emph{Instance:} A digraph $D=(V,A)$, each arc $ij$ has an integral positive
  weight $w_{ij}$, and a positive integer~$k$.\nopagebreak

  \emph{Parameter:} The integer $k$.\nopagebreak

  \emph{Question:} Is there an acyclic subdigraph of $D$ of weight at least
  $W/2+k$, where $W=\sum_{ij\in A}w_{ij}$ ?
\end{quote}
Mahajan, Raman, and Sikdar~\cite{MahajanRamanSikdar09} asked whether
\textsc{LOALB} is fixed-parameter tractable for the special case when all arcs
are of weight 1 (i.e., $D$ is unweighted). In this section we will prove that
\textsc{LOALB} admits a kernel with $O(k^2)$ arcs; consequently the problem is
fixed-parameter tractable. Note that if we allow weights to be positive reals,
then we can show, similarly to the $\NP$-completeness proof given in the
next section, that \textsc{LOALB} is $\NP$-complete already for $k=1$.

Consider the following reduction rule:
\begin{krule}\label{LO1}
Assume $D$ has a directed
\hbox{2-cycle $iji$};
if $w_{ij}=w_{ji}$ delete the cycle,
if $w_{ij}>w_{ji}$ delete the arc $ji$ and replace $w_{ij}$ by $w_{ij}-w_{ji}$,
and if $w_{ji}>w_{ij}$ delete the arc $ij$ and replace $w_{ji}$ by $w_{ji}-w_{ij}$.
\end{krule}
It is easy to check that the answer to LOALB for a digraph $D$ is {\sc yes} if
and only if the answer to LOALB is {\sc yes} for a digraph obtained from $D$
using the reduction rule as long as possible.

Let $D=(V,A)$ be an oriented graph, let $n=|V|$ and $W=\sum_{ij\in
  A}w_{ij}$. Consider a random bijection: $\alpha: V \dom \{1,\ldots,n\}$ and a
random variable $X(\alpha)=\frac{1}{2}\sum_{ij\in A} \ep_{ij}(\alpha)$, where
$\ep_{ij}(\alpha)=w_{ij}$ if $\alpha(i)<\alpha(j)$ and
$\ep_{ij}(\alpha)=-w_{ij}$, otherwise. It is easy to see that $X(\alpha)=\sum\SB
w_{ij}\SM ij\in A, \alpha(i)<\alpha(j)\SE - W/2$. Thus, the answer to LOALB is
{\sc yes} if and only if there is a bijection $\alpha: V \dom \{1,\ldots,n\}$
such that $X(\alpha)\ge k$. Since $\mathbb{E}(\ep_{ij})=0$, we have
$\mathbb{E}(X)=0$.

Let $W^{(2)}=\sum_{ij\in A}w_{ij}^2$. We will prove the following:
\begin{lemma}\label{lemEX2}
$\mathbb{E}(X^2)\ge W^{(2)}/12$.
\end{lemma}
\begin{proof} Let $N^+(i)$ and $N^-(i)$ denote the sets of out-neighbors and
  in-neighbors of a vertex $i$ in $D$. By the definition of $X$,
  \begin{equation}\label{EX2eq} 4\cdot \mathbb{E}(X^2) = \sum_{ij\in A}
    \mathbb{E}(\ep_{ij}^2) + \sum_{ij,pq\in A}
    \mathbb{E}(\ep_{ij}\ep_{pq}),\end{equation} where the second sum is taken
  over ordered pairs of distinct arcs. Clearly, $\sum_{ij\in
      A} \mathbb{E}(\ep_{ij}^2)=W^{(2)}.$ To compute $\sum_{ij,pq\in A}
\mathbb{E}(\ep_{ij}\ep_{pq})$ we consider the following cases:

\Case{Case 1:} $\{i,j\}\cap \{p,q\}=\emptyset$. Then $\ep_{ij}$ and $\ep_{pq}$
are independent and
$\mathbb{E}(\ep_{ij}\ep_{pq})=\mathbb{E}(\ep_{ij})\mathbb{E}(\ep_{pq})=0$.

\Case{Case 2a:} $|\{i,j\}\cap \{p,q\}|=1$ and $i=p$. Since the probability that
$i<\min\{j,q\}$ or $i>\max\{j,q\}$ is $2/3$, $\ep_{ij}\ep_{iq}=w_{ij}w_{iq}$
with probability $\frac{2}{3}$ and $\ep_{ij}\ep_{iq}=-w_{ij}w_{iq}$ with probability
$\frac{1}{3}$. Thus, for every $i\in V$ we have $\sum_{ij,iq\in A}\mathbb{E}(\ep_{ij}\ep_{iq}) = \frac{1}{3}\sum\SB w_{ij}w_{iq} \SM
j\neq q\in N^+(i)\SE =\frac{1}{3} (\sum_{j\in N^+(i)}w_{ij}
)^2-\frac{1}{3}\sum_{j\in N^+(i)}w^2_{ij}$.

\Case{Case 2b:} $|\{i,j\}\cap \{p,q\}|=1$ and $j=q$. Similarly to Case 2a, we
obtain $\sum_{ij,pj\in A}\mathbb{E}(\ep_{ij}\ep_{pj})=\frac{1}{3} (\sum_{i\in
  N^-(j)}w_{ij})^2-\frac{1}{3}\sum_{i\in N^-(j)}w^2_{ij}$.

\Case{Case 3a:} $|\{i,j\}\cap \{p,q\}|=1$ and $i=q$. Since
$\ep_{ij}\ep_{pi}=w_{ij}w_{pi}$ with probability $\frac{1}{3}$ and
$\ep_{ij}\ep_{pi}=-w_{ij}w_{pi}$ with probability $\frac{2}{3}$, we obtain
$\sum_{ij,pi\in A}\mathbb{E}(\ep_{ij}\ep_{pi})=-\frac{1}{3}\sum\SB w_{ij}w_{pi} \SM j\in N^+(i),\
p\in N^-(i)\SE=-\frac{1}{3}\sum_{j\in N^+(i)}w_{ij}\sum_{p\in N^-(i)}w_{pi}$.

\Case{Case 3b:} $|\{i,j\}\cap \{p,q\}|=1$ and $j=p$. Similarly to Case 3a, we
obtain $\sum_{ij,jq\in A}\mathbb{E}(\ep_{ij}\ep_{jq})=-\frac{1}{3}\sum_{i\in
  N^-(j)}w_{ij}\sum_{q\in N^+(j)}w_{jq}$.

\medskip\noindent\sloppypar Equation (\ref{EX2eq}) and the subsequent computations imply that
$4\cdot \mathbb{E}(X^2)= W^{(2)}+ \frac{1}{3}(Q-R),$
where
\newcommand{\xleft}{\big} 
\newcommand{\xright}{\big} 
$$Q = \sum_{i\in  V} \left( \xleft(\sum_{j\in N^+(i)}w_{ij}  \xright)^2-\sum_{j\in
  N^+(i)}w^2_{ij}+\xleft(\sum_{j\in N^-(i)}w_{ji}\xright)^2-\sum_{j\in
  N^-(i)}w^2_{ji} \right),$$ and
$$R=2\cdot \sum_{i\in V}\xleft(\sum_{j\in N^+(i)}w_{ij}\xright)\xleft(\sum_{j\in
  N^-(i)}w_{ji}\xright).$$
By the inequality of arithmetic and geometric means, for each $i\in V$, we have
\[
\xleft(\sum_{j\in N^+(i)}w_{ij}\xright)^2 + \xleft(\sum_{j\in
    N^-(i)}w_{ji}\xright)^2 - 2\xleft(\sum_{j\in
    N^+(i)}w_{ij}\xright)\xleft(\sum_{j\in N^-(i)}w_{ji}\xright)\ge 0.
\]
Therefore,
\[
Q-R\ge -\sum_{i\in V}\sum_{j\in N^+(i)}w^2_{ij}-\sum_{i\in V}\sum_{j\in
  N^-(i)}w^2_{ji}=-2W^{(2)},
\]
and $4\cdot \mathbb{E}(X^2)\ge W^{(2)}-2W^{(2)}/3=W^{(2)}/3$, implying
$\mathbb{E}(X^2)\ge W^{(2)}/12$.
\end{proof}

Now we can prove the main result of this section.

\begin{theorem}\label{thmainMAS}
  The problem \textsc{LOALB} admits a kernel with $O(k^2)$ arcs.
\end{theorem}
\begin{proof}
Let $H$ be a digraph. We know that the answer to LOALB for $H$ is {\sc yes} if
and only if the answer to LOALB is {\sc yes} for a digraph $D$ obtained from $H$
using Reduction Rule \ref{LO1} as long as possible. Observe that $D$ is an oriented graph.
Let $\cal B$ be the set of bijections from $V$ to $\{1,\ldots,n\}$.
Observe that $f:\ {\cal B}\dom {\cal B}$ such that $f(\alpha(v))=|V|+1-\alpha(v)$ for each $\alpha\in{\cal B}$
is a bijection. Note that $X(f(\alpha))=-X(\alpha)$ for each $\alpha\in{\cal B}$. Therefore,
$\Prob(X=a)=\Prob(X=-a)$ for each real $a$ and, thus,
$X$ is symmetric. Thus, by Lemmas \ref{eqsym} and
  \ref{lemEX2}, we have $\Prob(\ X\ge \sqrt{W^{(2)}/12}\ )>0$. Hence, if
  $\sqrt{W^{(2)}/12}\ge k$, there is a bijection $\alpha: V\dom \{1,\ldots ,n\}$
  such that $X(\alpha)\ge k$ and, thus, the answer to LOALB (for both $D$ and $H$) is {\sc
    yes}. Otherwise, $|A|\le W^{(2)}<12\cdot k^2$.
\end{proof}

  \noindent We close this section by outlining how Theorem~\ref{thmainMAS} can be used to
  actually find a solution to \textsc{LOALB} if one exists.  Let $(D,k)$ be an
  instance of \textsc{LOALB} where $D=(V,A)$ is a directed graph with integral
  positive arc-weights and $k\geq 1$ is an integer.  Let $W$ be the total weight
  of $D$. As discussed above, we may assume that $D$ is an oriented graph.  If
  $\Card{A}<12k^2$ then we can find a solution, if one exists, by trying all
  subsets $A'\subseteq A$, and testing whether $(V,A')$ is acyclic and has
  weight at least $W/2+k$; this search can be carried out in time $2^{O(k^2)}$.
  Next we assume $\Card{A} \geq 12k^2$. We know by Theorem~\ref{thmainMAS} that
  $(D,k)$ is a \textsc{yes}-instance; it remains to find a solution.

  For a vertex $i\in V$ let $d_D(i)$ denote its unweighted degree in $D$, i.e.,
  the number of arcs (incoming or outgoing) that are incident with~$i$.
  Consider the following reduction rule:
  \begin{krule}\label{LO2} If there is a vertex $i\in V$ with
  $\Card{A}-12k^2 \geq d_D(i)$, then delete $i$ from~$D$.
  \end{krule}
  Observe that by applying the rule we obtain again a \textsc{yes}-instance $(D-i,k)$ of
  \textsc{LOALB} since $D-i$ has still at least $12k^2$ arcs.  Moreover, if we
  know a solution $D_i'$ of $(D-i,k)$, then we can efficiently obtain a solution
  $D'$ of $(D,k)$: if $\sum_{j\in N^+(i)} w_{ij} \geq \sum_{j\in N^-(i)} w_{ij}$
  then we add $i$ and all outgoing arcs $ij\in A$ to~$D_i'$; otherwise, we add
  $i$ and all incoming arcs $ji\in A$ to $D_i'$. After multiple applications of
  Rule \ref{LO2} we are left with an instance $(D_0,k)$ to which Rule \ref{LO2} cannot be
  applied.  Let $D_0=(V_0,A_0)$.  We pick a vertex $i\in V_0$. If $i$ has a
  neighbor $j$ with $d_{D_0}(j)= 1$, then $\Card{A_0} \leq 12k^2$, since
  $\Card{A_0}-d_{D_0}(j) < 12k^2$.  On the other hand, if $d_{D_0}(j)\geq 2$ for
  all neighbors $j$ of $i$, then $i$ has less than $2\cdot 12k^2$ neighbors,
  since $D_0-i$ has less than $12k^2$ arcs; thus $\Card{A_0}<3\cdot
  12k^2$. Therefore, as above, time $2^{O(k^2)}$ is sufficient to try all
  subsets $A_0'\subseteq A_0$ to find a solution to the instance $(D_0,k)$. Let
  $n$ denote the input size of instance~$(D,k)$. Rule \ref{LO2} can
  certainly be applied in polynomial time $n^{O(1)}$, and we apply it less than
  $n$ times. Hence, we can find a solution to $(D,k)$, if one exists, in time
  $n^{O(1)}+2^{O(k^2)}$.

Recall that a kernelization reduces in polynomial time an instance
$(I,k)$ of a parameterized problem to a
\emph{decision-equivalent} instance $(I',k')$, its \emph{problem
kernel}, where $k'\leq k$ and the size of $I'$ is bounded by a
function of~$k$. Solutions for $(I,k)$ and solutions for $(I',k')$ are
possibly unrelated to each other. We call $(I',k')$ a \emph{faithful
problem kernel} if from a solution for $(I',k')$ we can construct a
solution for $(I,k)$ in time polynomial in $|I|$ and $k$.  Clearly the above
$(D_0,k)$ is a faithful kernel.



\section{Max Lin-2}\label{secLin}

Consider a system of $m$ linear equations $e_1,\ldots, e_m$ in $n$
variables $z_1, \ldots ,z_n$ over $\GF(2)$, and suppose that each equation
$e_j$ has a positive integral weight $w_j$, $j=1,\ldots ,m$. The problem
\textsc{Max Lin-2} asks for an assignment of values to the variables that
maximizes the total weight of the satisfied equations. Let $W=w_1+\cdots +w_m$.

To see that the total weight of the equations that can be satisfied is at least
$W/2$, we describe a simple procedure suggested in \cite{HastadVenkatesh02}.
We assign values to the variables $z_1,\dots,z_n$ one by one and
simplify the system after each assignment. When we wish to assign 0 or 1 to
$z_i$, we consider all equations reduced to the form $z_i = b$, for a constant
$b$. Let $W'$ be the total weight of all such equations. We set $z_i:=0$, if the total
weight of such equations is at least $W'/2$, and set $z_i:=1$, otherwise.
If there are no equations of the form $z_i = b$, we set $z_i:=0.$
To see that the lower bound $W/2$ is tight, consider a system consisting of
pairs of equations of the form $\sum_{i\in I} z_i = 1$ and $\sum_{i\in I} z_i
= 0$ where both equations have the same weight.

The parameterized complexity of \textsc{Max Lin-2} parameterized above the tight
lower bound $W/2$ was stated by Mahajan, Raman and
Sikdar~\cite{MahajanRamanSikdar09} as an open question:
\begin{quote}
  \textsc{Max Lin-2 Parameterized Above Tight Lower Bound}
  (\textsc{LinALB})\nopagebreak

  \emph{Instance:} A system $S$ of $m$ linear equations $e_1,\ldots,e_m$ in $n$
  variables $z_1, \ldots ,z_n$ over $\GF(2)$, each equation $e_i$ with a
  positive integral weight $w_i$, $i=1,2,\ldots,m$, and a positive integer~$k$.
  Each equation $e_j$ can be written as $\sum_{i\in I_j}z_i=b_j$, where $\emptyset\neq I_j\subseteq \{1,\ldots ,n\}$.
    \nopagebreak

  \emph{Parameter:} The integer~$k$. \nopagebreak

  \emph{Question:} Is there an assignment of values to the variables
  $z_1,\dots,z_n$ such that the total weight of the satisfied
  equations is at least $W/2+k$, where $W=\sum_{i=1}^m w_i$?
\end{quote}
Let $r_j$ be the number of variables in equation $e_j$, and let
$r(S)=\max_{i=1}^m r_j$. We are not able to determine whether \textsc{LinALB} is
fixed-parameter tractable or not, but we can prove that the following three special cases are
fixed-parameter tractable: (1) there is a set $U$ of variables such that each equation contains an
odd number of variables from $U$, (2) there is a constant $r$
such that $r(S)\le r$, (3) there is a constant $\rho$
such that any variable appears in at most $\rho$ equations.

Notice that in our formulation of \textsc{LinALB} it is required that each
equation has a positive integral weight. In a relaxed setting in which an
equation may have any positive real number as its weight, the problem is
$\NP$-complete even for $k=1$ and each $r_j=2$. Indeed, let each linear equation
be of the form $z_u+z_v=1$. Then the problem is equivalent to \textsc{MaxCut},
the problem of finding a cut of total weight at least $L$ in an undirected graph
$G$, where $V(G)$ is the set of variables, $E(G)$ contains $(z_u,z_v)$ if and
only if there is a linear equation $z_u+z_v=1$, and the weight of an edge
$(z_u,z_v)$ equals the weight of the corresponding linear equation. The problem
\textsc{MaxCut} is a well-known $\NP$-complete problem.  Let us transform an
instance $I$ of \textsc{MaxCut} into an instance $I'$ of the ``relaxed''
\textsc{LinALB} by replacing the weight $w_i$ by $w'_i:=w_i/(L-W/2)$.  We may
assume that $L-W/2>0$ since otherwise the instance is immediately seen as a {\sc
  yes}-instance. Observe that the new instance $I'$ has an assignment of values
with total weight at least $W'/2+1$ if and only if $I$ has a cut with total
weight at least $L$.  We are done.

Let $A$ be the matrix of the coefficients of the variables in $S$.
It is well-known that the maximum number
of linearly independent columns of $A$
equals ${\rm rank} A$, and such a collection of columns can be found in time polynomial
in $n$ and $m$, using, e.g., the Gaussian elimination on columns \cite{Shores}. We have the
following reduction rule and supporting lemma.

\begin{krule}\label{rulerank}
Let $A$ be the matrix of the coefficients of the variables in $S$,
let $t={\rm rank} A$ and let columns $a^{i_1},\ldots ,a^{i_t}$ of $A$ be linearly independent.
Then delete all variables not in $\{z_{i_1},\ldots ,z_{i_t}\}$ from the equations of $S$.
\end{krule}

\begin{lemma}\label{mGEn}
Let $T$ be obtained from $S$ by Rule~\ref{rulerank}.
Then $T$ is a {\sc yes}-instance if and only if $S$ is a {\sc yes}-instance.
Moreover, $T$ can be obtained from $S$  in time polynomial in $n$ and $m$.
\end{lemma}
\begin{proof}
If $t=n$, set $T:=S$, so assume that $t<n.$ The remark before the lemma immediately implies that $T$ can be obtained
from $S$  in time polynomial in $n$ and $m$.
Let $S'$ be a system of equations from $S$ and
let $T'$ be the corresponding system of equations from $T$.
It is sufficient to prove the following claim:

{\em There is an assignment of values to $z_1,\ldots ,z_n$ satisfying all equations in $S'$ and falsifying
the rest of equations in $S$ if and only if there is an  assignment of values to $z_{i_1},\ldots ,z_{i_t}$ satisfying all equations in $T'$  and falsifying
the rest of equations in $T$.}

Let an assignment $z^0$ of values to $z=(z_1,\ldots ,z_n)$ satisfy all equations of  $S'$ and falsify the equations of $S''$, where $S''=S\setminus S'.$
This assignment satisfies all equations of $R$, the system obtained from $S$ by replacing the right hand side $b_j$ of each equation in $S''$ by
$1-b_j.$ Note that $R$ has the same matrix $A$ of coefficients with columns $a^1,\ldots,a^n.$
Let a column $a^i\not\in \{a^{i_1},\ldots ,a^{i_t}\}.$ Then,
by definition of $a^{i_1},\ldots ,a^{i_t}$, $a^i=\lambda_1a^{i_1}+\cdots +\lambda_ta^{i_t}$
for some numbers $\lambda_j\in \{0,1\}.$
Knowing the numbers  $\lambda_j$, we may eliminate a variable $z_i$ from $R$ by replacing $a^i$ with the sum of all columns
from $\{a^{i_1},\ldots ,a^{i_t}\}$ for which $\lambda_j=1$ and carrying out the obvious simplification of the system.
Thus, we may eliminate from $R$ all variables $z_i\not\in \{z_{i_1},\ldots ,z_{i_t}\}$ and get $y_{i_1}a^{i_1}+\cdots +y_{i_t}a^{i_t}=b',$ where
$b'$ is the right hand side of $R$ and each $y_j\in \{0,1\}.$ Now replace, in the modified $R$, the right hand side $b'_j$ of each
equation corresponding to an equation in $S''$ by $1-b'_j$ obtaining $T$.
Clearly, $(y_{i_1},\ldots ,y_{i_t})$ satisfies all equations of $T'$ and falsifies all equations
in $T''=T\setminus T'.$

Suppose now that $(y_{i_1},\ldots ,y_{i_t})$ satisfies all equations of $T'$ and falsifies all equations
in $T''$. Then $(y_1,\ldots ,y_n)$, where $y_j=0$ if $j\not\in \{i_1,\ldots ,i_t\}$, satisfies all equations of $S'$ and falsifies all equations
in $S''$. Thus, the claim has been proved.
\end{proof}

Consider the following reduction rule for \textsc{LinALB} used in \cite{HastadVenkatesh02}.
\begin{krule}\label{rule1}
If we have, for a subset $I$ of $\{1,2,\ldots ,n\}$, the equation $\sum_{i \in I} z_i =b'$
with weight $w'$, and the equation $\sum_{i \in I} z_i =b''$ with weight $w''$,
then we replace this pair by one of these equations with weight $w'+w''$ if $b'=b''$ and, otherwise, by
the equation whose weight is bigger, modifying its
new weight to be the difference of the two old ones. If the resulting weight
is~0, we omit the equation from the system.
\end{krule}

If Rule \ref{rule1} is not applicable to a system we call the system {\em reduced under Rule \ref{rule1}}. Note that the
problem \textsc{LinALB} for $S$ and the system obtained from $S$ by applying Rule \ref{rule1} as long as possible have the
same answer.

Let $I_j \subseteq \{1, \ldots , n\}$ be the set of indices of the variables
participating in equation $e_j$, and let $b_j \in \{0,1\}$ be the right hand
side of $e_j$. Define a random variable $X=\sum_{j=1}^m X_j$, where
$X_j=(-1)^{b_j} w_j \prod_{i \in I_j} \ep_i$ and all the $\ep_i$ are independent
uniform random variables on $\{-1,1\}$ ($X$ was first introduced
in~\cite{HastadVenkatesh02}). We set $z_i=0$ if $\ep_i=1$ and $z_i=1$, otherwise, for each $i$. Observe that
$X_j=w_j$ if $e_j$ is satisfied and $X_j=-w_j$, otherwise.

\begin{lemma}\label{Xlin}
Let $S$ be reduced under Rule \ref{rule1}. The weight of the satisfied equations is at
least $W/2+k$ if and only if $X\ge 2k$. We have $\mathbb{E}(X) =0$ and $\mathbb{E}(X^2)=\sum_{j=1}^m
w^2_j.$
\end{lemma}
\begin{proof}
Observe that $X$ is the
difference between the weights of satisfied and non-satisfied equations.
Therefore, the weight of the satisfied equations equals $(X+W)/2$, and it is at
least $W/2+k$ if and only if $X\ge 2k$.
Since $\ep_i$ are independent, $\mathbb{E}(\prod_{i \in I_j} \ep_i)=\prod_{i \in I_j}\mathbb{E}( \ep_i)=0$. Thus,
$\mathbb{E}(X_j) =0$ and $\mathbb{E}(X) =0$ by linearity of expectation. Moreover,
\[
\mathbb{E}(X^2) = \sum_{j=1}^m
\mathbb{E}(X_j^2)+\sum_{1\le j\neq q\le m} \mathbb{E}(X_jX_q)=\sum_{j=1}^m
w^2_j>0
\] as  $\mathbb{E}(\prod_{i \in I_j} \ep_i\cdot \prod_{i \in I_q} \ep_i)=\mathbb{E}(\prod_{i \in I_j\Delta I_q} \ep_i)=0$
implies $\mathbb{E}(X_jX_q)=0,$ where $I_j\Delta I_q$ is the symmetric difference between $I_j$ and $I_q$
($I_j\Delta I_q\neq \emptyset$ due to Reduction Rule \ref{rule1}).
\end{proof}

\begin{lemma}\label{EX4}
Let $S$ be reduced under Rule \ref{rule1} and suppose that no variable appears in more than $\rho\ge 2$ equations of $S$. Then $\mathbb{E}(X^4)\le 2\rho^2(\mathbb{E}(X^2))^2.$
\end{lemma}
\begin{proof}
Observe that \begin{equation}\label{x4eq}\mathbb{E}(X^4)=\sum_{(p,q,s,t)\in [m]^4}\mathbb{E}(X_pX_qX_sX_t),\end{equation}
where $[m]=\{1,\ldots ,m\}.$
Note that if the product $X_pX_qX_sX_t$ contains a variable $\ep_i$ in only one or three of the factors, then $\mathbb{E}(X_pX_qX_sX_t)=A\cdot \mathbb{E}(\ep_i)=0,$
where $A$ is a polynomial in random variables $\ep_l$, $l\in \{1,\ldots ,n\}\setminus \{i\}.$
Thus, the only nonzero terms in (\ref{x4eq}) are those for which either (1) $p=q=s=t$, or (2) there are two distinct integers $j,l$ such that each of them coincides with two elements in the sequence $p,q,s,t$, or (3) $|\{p,q,s,t\}|=4$, but each variable $\ep_i$ appears in an even number of the factors in $X_pX_qX_sX_t$. In Cases 1 and 2, we have $\mathbb{E}(X_pX_qX_sX_t)=w^4_p$ and  $\mathbb{E}(X_pX_qX_sX_t)=w_j^2w^2_l,$ respectively.
In Case 3,  $$\mathbb{E}(X_pX_qX_sX_t)\le  w_pw_qw_sw_t\le (w_p^2w_q^2+w_s^2w_t^2)/2.$$

Let $1\le j< l\le m$. Observe that $\mathbb{E}(X_pX_qX_sX_t)=w_j^2w^2_l$ in Case 2 for ${4 \choose 2}=6$
4-tuples $(p,q,s,t)\in [m]^4.$ In Case 3, we claim that $j,l\in \{p,q,s,t\}$ for at most $4\cdot (\rho-1)^2$
4-tuples $(p,q,s,t)\in [m]^4$.  To see this, first note that
$w_p^2w_q^2$ and $w_s^2w_t^2$ appear in our upper bound on
$\mathbb{E}(X_pX_qX_sX_t)$ (with coefficient 1/2).
Therefore, there are only four possible ways for $w_j^2w_l^2$ to appear in our upper bound, namely the following: (i) $j=p, l=q$,
(ii) $l=p, j=q$, (iii) $j=s, l=t$, and (iv) $l=s, j=t.$
Now assume, without loss of generality, that $j=p$ and $l=q$.
Since $S$ is reduced under Rule \ref{rule1}, the product $X_jX_l$ must have a variable $\ep_i$ of degree one. Thus, $\ep_i$ must be in $X_s$ or $X_t$,
but not in both (two choices). Assume that $\ep_i$ is in $X_s$.
Observe that there are at most $\rho-1$ choices for $s$. Note that $X_jX_lX_s$
must contain contain a variable $\ep_{i'}$ of odd degree.
Thus, $\ep_{i'}$ must be in $X_t$ and, hence, there are at most $\rho-1$ choices for $t$.

Therefore, we have $$\mathbb{E}(X^4)\le \sum_{j=1}^m w^4_j + (6+4(\rho-1)^2)\sum_{1\le j< l\le m}w_j^2w_l^2<2\rho^2\left(\sum_{j=1}^m
w^2_j\right)^2.$$ Thus, by Lemma \ref{Xlin}, $\mathbb{E}(X^4)\le 2\rho^2(\mathbb{E}(X^2))^2.$
\end{proof}

Case 1 of Theorem \ref{thLin} is of interest since its condition can be checked in polynomial time due to the following:

\begin{proposition}
We can check, in polynomial time, whether there exists a set $U$ of
variables such that each equation of $S$ contains an odd number of variables from $U$.
\end{proposition}
\begin{proof}
Observe that such a set $U$ exists if and only if the unweighted system $S'$ of linear
equations over $\GF(2)$ obtained from $S$ by replacing each $b_j$ with 1 has a solution. Indeed, if $U$ exists, set
$z_j=1$ for each $z_j\in U$ and $z_j=0$ for each $z_j\not\in U$. This assignment is a solution to $S'.$ If a solution to $S'$ exists,
form $U$ by including in it all variables $z_j$ which equal 1 in the solution. We can check whether $S'$ has a solution using the Gaussian elimination
or other polynomial-time algorithms, see, e.g., \cite{Copper93}.
\end{proof}

Now we can prove the following:
\begin{theorem}\label{thLin} 
Let $S$ be reduced under Rule \ref{rule1}. The following three special cases of \textsc{LinALB} are fixed-parameter tractable: (1) there is a set $U$ of
variables such that each equation contains an odd number of variables from $U$, (2) there is a constant $r$ such that $r(S)\le r$, (3) there is a constant $\rho$,
such that any variable appears in at most $\rho$ equations.
In each case, there exists a kernel with $O(k^2)$
equations and variables.
\end{theorem}
\begin{proof}
{\bf Case 1.} Let $z^0=(z^0_1,\ldots,z^0_n)\in \{0,1\}^n$ be an assignment of values to the variables
$z_1,\ldots,z_n,$ and let $-z^0=(z'_1,\ldots,z'_n),$ where $z'_i=1-z^0_i$ if $z_i\in U$ and  $z'_i=z^0_i,$ otherwise, $i=1,\ldots ,n.$
Observe that $f:\ z^0\mapsto -z^0$ is a bijection on the set of assignments and $X(-z^0)=-X(z^0)$. Thus, $X$ is a symmetric random variable.
Therefore, by Lemmas \ref{eqsym} and  \ref{Xlin},  $\Prob(\ X\ge \sqrt{m}\ )\ge \Prob(\ X\ge \sqrt{\sum_{j=1}^m
w^2_j}\ )>0.$ Hence, if $\sqrt{m}\ge 2k$, the answer to \textsc{LinALB} is {\sc
yes}. Otherwise, $m<4k^2$ and after applying  Rule \ref{rulerank}, we obtain a kernel with $O(k^2)$
equations and variables.

\2

\noindent{\bf Case 2.} Since $X$ is a polynomial of degree at most $r$, it follows by Lemma \ref{lem41}
that $\mathbb{E}(X^4) \leq 2^{6r} \mathbb{E}(X^2)^2$. This inequality and the
results in the previous paragraph show that the conditions of Lemma \ref{lem32}
are satisfied and, thus,
\[\Prob\!\left(X > \frac{\sqrt{\sum_{j=1}^mw_j^2}}{4 \cdot
    8^r}\right)>0, \quad \text{implying} \quad \Prob\!\left(X >
  \frac{\sqrt{m}}{4 \cdot 8^r}\right)>0.
\]
Consequently, if $2k-1\le \sqrt{m}/(4 \cdot 8^r)$, then there is an assignment
of values to the variables $z_1,\dots,z_n$ which satisfies equations of total
weight at least $W/2+k$. Otherwise, $2k-1> \sqrt{m}/(4 \cdot 8^r)$ and
$m<16(2k-1)^264^r$. After applying  Rule \ref{rulerank}, we obtain the required kernel.

\2

\noindent{\bf Case 3.} If $\rho=1$, it is easy to find an assignment to the variables that satisfies all equations of $S$.
Thus, we may assume that $\rho\ge 2.$ To prove that there exists a kernel with  $O(k^2)$ equations, we can proceed as in Case 2,
but use Lemma \ref{EX4} rather than Lemma \ref{lem41}.
\end{proof}

\noindent{\bf Remark 1.} Note that even if $S$ does not satisfy Case 2 of the theorem, $T$,
the system obtained from $S$ using Rule \ref{rulerank},
may still satisfy Case 2. However, we have not formulated the theorem for $S$ reduced under Rule \ref{rulerank} as the reduced
system depends on the choice of a maximum linear independent collection of columns of $A$.

\2

\noindent{\bf Remark 2.}
Unfortunately, we cannot use Lemma \ref{lem32} for $X$ to prove that the whole problem \textsc{LinALB} is fixed-parameter tractable.
This is due to the fact that $\mathbb{E}(X^4) \geq \Theta(\mathbb{E}(X^2)^3)$ for some systems $S$. One such system is $\sum_{i \in I} z_i =1$
for each nonempty subset $I$ of $\{1,\ldots ,n\}.$ (Thus, $m=2^n-1$.) We set the weight $w_j=1$ for every $j=1,\ldots ,m.$ By Lemma \ref{Xlin},
$\mathbb{E}(X^2)=m$. Let $Q_n$ be the set of 4-tuples $(p,q,s,t)\in [m]^4$, where $[m]=\{1,\ldots ,m\}$,
such that $|\{p,q,s,t\}|=4$ and each variable $\ep_i$ appears in an even
number of the factors in $X_pX_qX_sX_t$. By the proof of Lemma \ref{EX4}, we have $\mathbb{E}(X^4)> \sum_{(p,q,s,t)\in Q}X_pX_qX_sX_t=|Q_n|$.
Let $(p,q,s,t)\in Q_{n-1}$ (we allow only variables $\ep_1,\ldots ,\ep_{n-1}$ in $X_p,X_q,X_s$ and $X_t$).
We can construct eight 4-tuples in $Q_n$ using $(p,q,s,t)$ as
a starting point. Indeed, we can multiply each of $X_p,X_q,X_s$ and $X_t$ by $\ep_n$, or exactly two of $X_p,X_q,X_s$ and $X_t$ by $\ep_n$, or keep
$X_p,X_q,X_s$ and $X_t$ without a change. Thus, $|Q_n|\ge 8|Q_{n-1}|$. Therefore, $\mathbb{E}(X^4)>|Q_{n-3}|\ge 8^{n-3}>m^3/8^3=\mathbb{E}(X^2)^3/8^3$
as $|Q_3|\ge 1.$

\2

A faithful kernel can be found for all three cases of Theorem \ref{thLin},
but we restrict ourselves to Case 2.
Consider an instance $(S,k)$ of \textsc{LinALB} over $n$ variables with
$r=r(S)=O(1)$, $\Card{S}=m\geq f(k,r)$, where $f(k,r)=16(2k-1)^264^r$ and assume $S$ is reduced under Reduction
Rule \ref{rule1}.  We can find a solution
for $(S,k)$ (i.e., an assignment satisfying at least $W/2+k$ equations) in
polynomial time by using the observation from \cite{AlonGutinKrivelevich04} that
the random variables $\epsilon_i$ are $4r$-wise independent and, thus, one can
use an $O(n^{2r})$-size sample space to support each $\epsilon_i$ (for more
details and the sample space construction, see \cite{AlonGutinKrivelevich04}).

Alternatively, we can use a modification of the approach given earlier for
\textsc{LOALB} to obtain in polynomial time a faithful kernel. Similar approaches can be used to
obtain faithful kernels for Cases 1 and 3.
Consider the following reduction rule:

\begin{krule}\label{rule3} If there is some variable $x$ of $S$ that occurs in at most
$|S|-f(k,r)$ equations, then remove all the equations from $S$ in which $x$
occurs.\end{krule}
Let $T$ be the system obtained from $S$ using Rule \ref{rulerank}.
Apply Rule \ref{rule3} to $T$ as long as
possible. Note that we can transform
a solution for $(T,k)$ to a solution for $(S,k)$ by assigning zero to all variables not in $T$.
To show that the number of
equations in $T$ is polynomially bounded in $k$, let $x$ be a variable of $T$
that occurs in the smallest number of equations of $T$, and let $T=T_1\cup
T_2$ where all equations in $T_1$ contain $x$ and no equation in $T_2$
contains~$x$. Since $T$ is reduced under Rule \ref{rule3}, $\Card{T_2}<f(k,r)=O(k^2)$;
thus $T_2$ involves at most $r \cdot f(k,r)=O(k^2)$ variables.  However,  by Rule \ref{rulerank}
and the choice of $x$, $x$ is the only
variable that occurs in $T_1$ but does not occur in $T_2$.  Hence
$\Card{T_1}= O( (r \cdot f(k,r))^{r-1})=O(k^{2r-2})$ and $|T|=O(k^2+k^{2r-2}).$

\section{Max Exact $r$-SAT}\label{secSAT}

Let $x=(x_1, x_2, \ldots ,x_n)$ be a vector of Boolean variables. We assume that
each variable $x_i$ attains $1$ or $-1$ (meaning \textsc{true} and
\textsc{false}, respectively). We will denote the negation of a variable $x_i$
by $-x_i$ and we let $L=\{x_1,x_2,\dots,$\hskip0pt $x_n, -x_1, -x_2,
\dots,$\hskip0pt $-x_n\}$ be the set of literals over $x_1,\dots,x_n$. Let $r\ge
2$ be a fixed integer, and let $C=\{C_1, C_2, \ldots, C_m\}$ be a set of
clauses, each involving exactly $r$ distinct literals from $L$ such that for no pair
$y,z$ of literals we have $y=-z$. Then $\CCC$ is an \emph{exact $r$-CNF
formula}.

Consider an exact $r$-CNF formula $\CCC$ with $m$ clauses and a random truth
assignment for $x$. Since the probability of a clause of $\CCC$ to be satisfied
is $1-2^{-r}$, the expected number of satisfied clauses in $\CCC$ is
$(1-2^{-r})m$. Thus, there is a truth assignment for $x$ that satisfies at least
$(1-2^{-r})m$ clauses.  This bound is tight as can be seen by considering an
$r$-CNF formula that contains all $2^r$ possible clauses over the same
$r$ variables, or by considering a disjoint union of several such $r$-CNF
formulas.


Mahajan, Raman and Sikdar~\cite{MahajanRamanSikdar09} stated the complexity of
the following problem as an open question.

\begin{quote}
  \textsc{Exact $r$-SAT Above Tight Lower Bound}
  (\textsc{$r$-SATALB})\nopagebreak

  \emph{Instance:} An exact $r$-CNF formula $\CCC$ and a positive rational
  number~$k$ with denominator $2^r$.\nopagebreak

  \emph{Parameter:} The number~$k$.\nopagebreak

  \emph{Question:} Is there a truth assignment $(x'_1,x'_2, \ldots ,x'_n)$
  satisfying at least $(1-2^{-r})m+k$ clauses of $\CCC$?
\end{quote}

Mahajan, Raman and Sikdar~\cite{MahajanRamanSikdar09} require $k$
  to be a positive integer, but since $(1-2^{-r})m$ is a positive rational
  number with denominator $2^r$, our setting for $k$ seems more
  natural.

We will prove that the problem has a quadratic kernel for a wide family of
instances.
We say that a pair of distinct
clauses $Y$ and $Z$ has a {\em conflict} if there is a literal $p\in Y$ such
that $-p\in Z$. We say that a pair $Y$ and $Z$ has an {\em overlap} if $Y$ and $Z$
have common literals, but they do not have a conflict. For an exact $r$-CNF formula $\CCC$, the {\em conflict number}
${\rm cn}({\CCC})=c-o$, where $c$ ($o$, respectively) is the number of ordered pairs of clauses having a conflict (an overlap, respectively).

\begin{theorem}\label{thSAT}
The problem \textsc{$r$-SATALB} restricted to exact $r$-CNF formulas with $m$
clauses and of conflict number at most $(2^{r}-2)m$ admits a quadratic kernel.
\end{theorem}

To establish the theorem, consider an exact $r$-CNF formula $\CCC$ with $m$
clauses and of conflict number at most $(2^{r}-2)m$, and
consider a clause $Z$ of $\CCC$. Let $x_{i_1}, x_{i_2}, \ldots , x_{i_r}$ be the
variables corresponding to the literals of $Z$ and let
$x^0_{i_1}, x^0_{i_2}, \ldots , x^0_{i_r}$ be the unique truth assignment not
satisfying $Z$. Define a random variable $X_Z$ as follows:
Let $V= \{-1,1\}^r-\{(x^0_{i_1}, x^0_{i_2}, \ldots , x^0_{i_r})\}$ and
\[
X_Z(x_1, x_2, \ldots ,x_n) =- \frac{2^r-1}{2^r}+\sum_{(v_1, \ldots ,v_r) \in V}
\frac{\prod_{j=1}^r (1+x_{i_j} v_j)}{2^r} .
\] The value of $X_Z$ at a truth assignment $x'=(x'_1,x'_2,
  \ldots ,x'_n)$ is $2^{-r}$ if $x'$ satisfies $Z$ (in this case the product equals $2^r$), and it is
$2^{-r}-1$, otherwise. Let $X =\sum_{Z \in \CCC} X_Z$ ($X$ is a special case of a random variable introduced in \cite{AlonGutinKrivelevich04}).

\medskip\noindent
We study some properties of $X$ in the following two lemmas.

\begin{lemma}\label{lemsat}
  Let $x'=(x'_1,x'_2, \ldots ,x'_n)$ be a truth assignment.  Then the value of
  $X$ at $x'$ equals $m(x')-(1-2^{-r})m$, where $m(x')$ is the number of clauses
  in $\CCC$ satisfied by $x'$. Thus, the answer to \textsc{$r$-SATALB} is {\sc
    yes} if and only if $X(x'')\ge k$ for some truth assignment $x''$.  We also
  have $\mathbb{E}(X) =0$ and
  $\mathbb{E}(X^4)\leq 2^{6r} \mathbb{E}(X^2)^2$.
\end{lemma}
\begin{proof} Let $x'=(x'_1,x'_2,
  \ldots ,x'_n)$ be a truth assignment. Observe that
  $X(x')=m(x')2^{-r}+(2^{-r}-1)(m-m(x'))=m(x')-(1-2^{-r})m$. Hence,
  $m(x')\ge (1-2^{-r})m + k$ if and only if $X(x')\ge k$.

  Observe that the probability of $Z$ being satisfied (not satisfied) is
  $1-2^{-r}$ ($2^{-r}$). Thus, the expectation of $X_Z$ is zero, and
  $\mathbb{E}(X) =0$ by linearity of expectation.
  Since $X$ is a polynomial of degree at most $r$ in $x_1, x_2, \ldots
  ,x_n$, it follows, by Lemma \ref{lem41}, that $\mathbb{E}(X^4) \leq 2^{6r}
  \mathbb{E}(X^2)^2$.
\end{proof}


\begin{lemma}\label{lemsigma}
  We have $\mathbb{E}(X^2) \ge m4^{-r}$.
\end{lemma}
\begin{proof}
  Observe that $\mathbb{E}(X^2)=\sum_{Z\in {\cal
      C}}\mathbb{E}(X_Z^2)+\sum_{Y\neq Z\in \CCC}\mathbb{E}(X_YX_Z)$.  We will
  compute $\mathbb{E}(X_Z^2)$ and $\mathbb{E}(X_YX_Z)$ separately.

  By the proof of Lemma \ref{lemsat}, $X_Z$ equals $2^{-r}$ with probability
  $1-2^{-r}$ and $2^{-r}-1$ with probability $2^{-r}$. Thus, $X^2_Z$ equals
  $2^{-2r}$ with probability $1-2^{-r}$ and $(2^{-r}-1)^2$ with probability
  $2^{-r}$. Hence, $\mathbb{E}(X_Z^2)=2^{-r}-4^{-r}$.

For a clause $Y$ of $\CCC$, let $\vars(Y)$ denote the sets of variables in $Y$
and let $\lits(Y)$ be the set of literals in $Y$.
To evaluate $\mathbb{E}(X_YX_Z)$, we consider the following three cases:

  \Case{Case 1:} $\vars(Y)\cap \vars(Z)=\emptyset$. Then $X_Y$ and $X_Z$ are
  independent random variables and, thus,
  $\mathbb{E}(X_YX_Z)=\mathbb{E}(X_Y)\mathbb{E}(X_Z)=0$.

  \Case{Case 2:} $Y$ and $Z$ have a conflict. Then $X_YX_Z$ equals $2^{-r}(2^{-r}-1)$ with probability
  $2^{-r+1}$, and $X_YX_Z$ equals $2^{-2r}$ with probability $1-2^{-r+1}$.
  Hence, $\mathbb{E}(X_YX_Z)=-4^{-r}$.

  \Case{Case 3:} $|\vars(Y)\cap \vars(Z)|=t>0$ and $|\lits(Y)\cap
  \lits(Z)|=t$. Since $Y\neq Z$, we have $1\le t<r$. Without loss of generality,
  assume that $\lits(Y)=\{x_1,\ldots,x_t,$\hskip0pt $x_{t+1},\ldots x_r\}$ and
  $\lits(Z)=\{x_1,\ldots ,x_t,x_{r+1},\ldots x_{2r-t}\}$. Thus, $X_YX_Z$ equals
  $(2^{-r}-1)^2$ with probability $2^{t-2r}$, $2^{-r}(2^{-r}-1)$ with
  probability $(2^{r-t+1}-2)/2^{2r-t}$, and $2^{-2r}$ with probability
  $(1-2^{r-t+1}+2^{2r-t})/2^{2r-t}$. Hence,

  $$\mathbb{E}(X_YX_Z)=2^{t-2r}(1-2^{-t})\ge 4^{-r}.$$

 Since ${\rm  cn}(\CCC)\le (2^r-2)m$, we have
  \[
  \label{e1}\mathbb{E}(X^2)\ge
  \sum_{Z\in  \CCC}\mathbb{E}(X_Z^2)+
  \!\!\!\sum_{Y\neq Z\in    \CCC}\!\!\mathbb{E}(X_YX_Z)\ge
  (2^{-r}-4^{-r})m -
  {\rm  cn}(\CCC)\cdot 4^{-r}\ge m4^{-r}.
  \]
  \end{proof}

\medskip\noindent Now we can complete the proof of Theorem \ref{thSAT}. By
Lemmas~\ref{lem32}, \ref{lemsat} and \ref{lemsigma}, $\Prob(\ X>
\sqrt{m}/(2^r\cdot 4\cdot 8^r)\ )>0$. Thus, if $\sqrt{m}/(2^r\cdot 4\cdot
8^r)\ge k$, there is a truth assignment $x'$ such that $X(x')\ge k$, i.e., the
answer to the instance of \textsc{$r$-SATALB} is {\sc yes}. Otherwise,
$m<16\cdot 256^rk^2$.  Thus Theorem \ref{thSAT} is established.

\medskip\noindent Consider a {\sc yes}-instance $(\CCC,k)$ of $r$-SATALB with
$n$ variables, $m\geq 16\cdot 256^rk^2$ clauses, and ${\rm  cn}(\CCC)\le (2^{r}-2)m$.
As in the previous section, we can find in polynomial time a
solution for $(\CCC,k)$ using the facts that each random variable $x_i$ is
$4r$-wise independent and there is an $O(n^{2r})$-size sample space to support
each~$x_i$.

\section{Discussions}\label{secd}

We have showed that the new method allows us to prove that some
maximization problems parameterized above tight lower bounds are
fixed-parameter tractable. Our method can also be used
for minimization problems parameterized below tight upper bounds.
As a simple example, consider the feedback arc problem: given a digraph $D=(V,A)$ find a minimum
set $F$ of arcs such that $D-F$ is acyclic. Certainly, $|A|/2$ is a tight upper bound on a minimum feedback set
and we can consider the parameterized problem which
asks whether $D$ has a feedback arc set with at most $|A|/2-k$ arcs.
Fixed-parameter tractability of this parameterized problem follows immediately from fixed-parameter
tractability of LOALB, but we could prove this result directly using essentially the same approach as for LOALB.

It would be interesting to obtain applications of our method to other problems parameterized
above tight lower bounds or below tight upper bounds. One such very recent application is given in \cite{GutinKimMnichYeo},
where an open problem due to Benny Chor and described in \cite{Niedermeier06} was solved. The random variable $X$ considered there
is not symmetric and both application of Lemma \ref{lem41} and computation of  $\mathbb{E}(X^2)$ are more involved than for the problems
considered in this paper.

Let us provide further comments on the problems considered in this paper. First, it is natural to parameterize \textsc{Max Lin-2}
not just by $k$ but also by $r=r(S)$. The proof of Case 2 of Theorem \ref{thLin} shows immediately that this two-parameter problem is
fixed-parameter tractable (but our problem kernel is no longer of polynomial size). This result can be viewed as a contribution
towards Multivariate Algorithmics as outlined by Fellows \cite{FelIWOCA}. Second, we have managed to obtain a polynomial problem kernel
for \textsc{$2$-SATALB} in \cite{GutinKimSzeiderYeo09}. The approach there is very different from the method we introduced here and it
involves signed weighted graphs, graph matching theory
and the first moment probabilistic method. Perhaps, the approach of \cite{GutinKimSzeiderYeo09} can be extended to \textsc{$3$-SATALB}, but we
doubt that it can be extended to \textsc{$r$-SATALB} for $r>3.$

\vspace{2mm}

\noindent{\bf Acknowledgments.} Research of Gutin, Kim and Yeo was
supported in part by an EPSRC grant.


\end{document}